\documentclass{article}
\usepackage[utf8]{inputenc} 
\usepackage[T1]{fontenc}
\usepackage{url}              
\usepackage{cite}             

\usepackage[cmex10]{amsmath}  
\usepackage{mleftright}       
\mleftright                   

\usepackage{graphicx}         
\usepackage{booktabs}         
\usepackage{mathdots}

\usepackage[utf8]{inputenc} 
\usepackage[T1]{fontenc}
\usepackage{url}
\usepackage{ifthen}
\usepackage{cite}
\usepackage{algorithm}
\usepackage{algpseudocode}
\usepackage[cmex10]{amsmath} 
\usepackage{amssymb}
\usepackage{xcolor}

\usepackage{subcaption}
\usepackage{comment}
\usepackage{tabularx}
\usepackage{enumitem}
\setlist{nolistsep}
\usepackage{mathtools}

\usepackage{tikz}
\usetikzlibrary {arrows.meta,automata,positioning} 

\usepackage{amsfonts, amsthm}

\makeatletter
\newcommand{\multiline}[1]{%
  \begin{tabularx}{\dimexpr\linewidth-\ALG@thistlm}[t]{@{}X@{}}
    #1
  \end{tabularx}
}
\makeatother

\newtheorem{theorem}{Theorem}

\theoremstyle{definition}
\newtheorem{definition}{Definition}
\theoremstyle{remark}

\newtheorem{example}{Example}

\title{An Improved Viterbi Algorithm for a Class of Optimal Binary Convolutional Codes}
\author{Zita Abreu, Julia Lieb, Michael Schaller}
\date{}

\begin{document}

\maketitle
\begin{abstract}
The most famous error-decoding algorithm for convolutional codes is the Viterbi algorithm. In this paper, we present a new reduced complexity version of this algorithm which can be applied to a class of binary convolutional codes with optimum column distances called $k$-partial simplex convolutional codes.
\end{abstract}

\section{Introduction}
The classical way of dealing with errors during data transmission over some communication channel have been linear block codes, which are vector spaces over some finite field $\mathbb F_q$. Convolutional codes as modules over $\mathbb F_q[z]$ are a generalization of linear block codes to the polynomial setting. These codes are 
often used in erasure channels such as the Internet; for an overview see \cite{johannesson1999, lieb2021}.
Also different error-decoding algorithms for convolutional codes have been developed over time, see e.g.\cite{lin2004error}. The most prominent of them is the Viterbi algorithm, introduced in \cite{viterbi}.
It is a maximum-likelihood algorithm, however it has the significant disadvantage that its complexity grows exponentially with the code memory.

Recently, in \cite{isit23}, a new construction of binary convolutional codes with optimal column distances was presented. These codes are very attractive as they are capable of correcting a maximal number of errors per time interval. 
However, nothing was said regarding the decoding of these codes.
They are built from partial simplex codes, which are a generalization of first-order Reed-Muller codes. This allows us to generalize a decoding algorithm for Reed-Muller codes and combine it with the Viterbi algorithm to obtain a new decoding algorithm targeted to these optimal codes and resulting in significantly reduced computational complexity compared to the original Viterbi algorithm.

The outline of the paper is as follows.
We provide definitions and basic results for convolutional codes in Section \ref{sec:Convcodes}. Section \ref{sec:DecodingofConvolutional Codes} describes the Viterbi algorithm. Section \ref{sec:DecodingofHadamard} explains how to decode first-order Reed-Muller codes, while Section \ref{sec:DecodingofPartialSimplexCodes} describes how this can be generalized to decode partial simplex codes. In Section \ref{sec:NewDecodingAlgorithm}, we describe the new algorithm for decoding the aforementioned optimal binary convolutional codes. Finally, in Section \ref{sec:ComplexityAnalysis}, we assess the complexity of the algorithm presented in the previous section.

\section{Convolutional Codes} \label{sec:Convcodes}

Denote by $\mathbb{F}_{q}[z]$ the ring of polynomials over the finite field with $q$ elements $\mathbb{F}_{q}$.

\begin{definition}
For $k,n\in\mathbb N$ with $k\leq n$, an $(n,k)$ \textbf{convolutional code} $\mathcal{C}$ is defined as $\mathbb{F}_{q}[z]$-submodule of $\mathbb{F}_{q}[z]^n$ of rank $k$. A matrix $G(z)\in \mathbb{F}_{q}[z]^{k \times n}$ 
such that
$$
\mathcal{C} 
= \{{v(z) \in \mathbb{F}_{q}[z]^{n}: v(z) = u(z)G(z) \text{ with } u(z) \in \mathbb{F}_{q}[z]^{k}\}.}\nonumber
$$
is called a \textbf{generator matrix} or \textbf{encoder} for $\mathcal{C}$.
The \textbf{degree} $\delta$ of $\mathcal{C}$ is defined as the maximal degree of the $k\times k$ minors of any generator matrix $G(z)$ of $\mathcal{C}$. We call $\mathcal{C}$ an $(n,k,\delta)$ convolutional code.
\end{definition}

\begin{definition}
Consider $G(z)=\sum_{i=0}^{\mu}G_i z^{i}\in\mathbb F_q[z]^{k\times n}$ with $G_{\mu} \neq 0$. For each $i$, $1\leq i\leq k$, the $i$-th \textbf{row degree} $\nu_i$ of $G(z)$ is defined as the largest degree of any entry in row $i$ of $G(z)$.
If the sum of row degrees of $G(z)$ is equal to the degree $\delta$ of the convolutional code generated by $G(z)$, then $G(z)$ is called a \textbf{minimal} generator matrix for $\mathcal{C}$ and
$\mu=\max_{i=1,\hdots,k}\nu_i$ is called the $\textbf{memory}$ of $\mathcal{C}$.
A minimal encoder $G(z)$ is said to have \textbf{generic row degrees} if  
$\mu=\lceil\frac{\delta}{k}\rceil$ and 
the last $k\lceil\frac{\delta}{k}\rceil-\delta$ rows of $G_{\mu}$ are zero.
  Denote by $\Tilde{G}_{\mu}\in\mathbb F_q^{(\delta+k-k\lceil\frac{\delta}{k}\rceil)\times n}$ the matrix consisting of the first $\delta+k-k\lceil\frac{\delta}{k}\rceil$, i.e. nonzero, rows of  $G_{\mu}$. Also denote by $\tilde{u}_i\in\mathbb F_q^{\delta+k-k\lceil\frac{\delta}{k}\rceil}$ the vector consisting of the first $\delta+k-k\lceil\frac{\delta}{k}\rceil$ components of $u_i\in\mathbb F_q^k$.
\end{definition}

In the following we introduce two main distance notions to measure the error-correcting capability of convolutional codes.

\begin{definition}
The weight 
of $v(z)=\sum_{t=0}^{\deg(v(z))}v_tz^t \in \mathbb{F}_q[z]^n$ is defined as $wt(v(z))=\sum_{t=0}^{\deg(v(z))}wt(v_t)$, where $wt(v_t)$ is the (Hamming) weight of $v_t\in\mathbb F_q^n$. Moreover we denote the (Hamming) distance by $d(v(z),\hat{v}(z))=wt(v(z)-\hat{v}(z))$ for $v(z),\hat{v}(z)\in\mathbb F_q[z]^n$.
The \textbf{free distance} of a convolutional code
  $\mathcal{C}$ is given by
  $$d_{free}(\mathcal{C}):=\min_{v(z)\in\mathcal{C}}\left\{wt(v(z))\
    |\ v(z) \neq 0\right\}.$$
\end{definition}

\begin{definition}
For $j\in\mathbb N_0$ and $v_{[0,j]}:=(v_0,\hdots,v_j)$, the \textbf{j-th column distance}\index{column distance} of a convolutional code $\mathcal{C}$ is defined as
\[
d_j^c(\mathcal{C}):=\min\left\{wt(v_{[0,j]})\ |\ {v}(z)\in\mathcal{C} \text{ and }{v}_0 \neq 0\right\}.
\]
\end{definition}

\begin{definition}\label{defopt}
   We say that a binary $(n,k,\delta)$ convolutional code $\mathcal{C}$ has \textbf{optimal column distances} if there exists no binary $(n,k,\delta)$ convolutional code $\hat{\mathcal{C}}$ such that $d^c_j(\hat{\mathcal{C}})>d^c_j(\mathcal{C})$ for some $j\in\mathbb N_0$ and $d^c_i(\hat{\mathcal{C}})=d^c_i(\mathcal{C})$ for all $0\leq i<j$.
\end{definition}

In the following we present a class of binary convolutional codes which are optimal in sense of the previous definition.

\begin{definition}[\hspace{-0.10mm}\cite{isit23}]
Let $S(\delta+k)\in\mathbb F_2^{(\delta+k)\times 2^{\delta+k}}$ 
be the generator matrix of a binary simplex code, i.e. the columns of $S(\delta+k)$ are exactly all elements from $\mathbb F_2^{\delta+k}\setminus\{0\}$,
and remove the columns whose first $k$ entries are equal to zero and define the resulting matrix as $S(k+\delta)_k\in\mathbb F_2^{(\delta+k)\times (2^{\delta+k}-2^{\delta})}$. We call the (block) code with generator matrix $S(\delta+k)_k$ a \textbf{$k$-partial simplex code} $\mathcal{S}(\delta+k)_k$ of dimension $\delta+k$. 
\end{definition}

Note that a $1$-partial simplex code of dimension $\delta+1$ is a Reed-Muller code $RM(1, \delta)$ of dimension $\delta+1$.

\begin{theorem}[\hspace{-0.10mm}\cite{isit23}]
Let $\mathcal{C}$ be a $(2^{\delta}(2^k-1),k,\delta)$ convolutional code with generator matrix $G(z)=\sum_{i=0}^{\lceil\frac{\delta}{k}\rceil}G_iz^i\in\mathbb F_2[z]^{k\times 2^{\delta}(2^k-1)}$ where $(G_0^\top\ \cdots\ \ G_{\mu-1}^\top \ \tilde{G}_{\mu}^\top)^\top=S(\delta+k)_k$. $\mathcal{C}$ is called \textbf{$k$-partial simplex convolutional code}. 
Then, 
$$d_j^c(\mathcal{C})=\begin{cases}
n\cdot\frac{2^{k-1}}{2^k-1}+j\frac{n}{2} & \text{for}\quad j\leq\lfloor\frac{\delta}{k}\rfloor\\
n\cdot\frac{2^{k-1}}{2^k-1}+\lfloor\frac{\delta}{k}\rfloor\cdot \frac{n}{2}  & \text{for}\quad j\geq\lfloor\frac{\delta}{k}\rfloor
\end{cases}$$
and these column distances are optimal in the sense of Definition \ref{defopt}.
Moreover, $$d_{free}(\mathcal{C})=\lim_{j\rightarrow\infty}d_j^c(\mathcal{C})=n\cdot\frac{2^{k-1}}{2^k-1}+\lfloor\frac{\delta}{k}\rfloor\cdot \frac{n}{2}.$$
\end{theorem}

\section{Viterbi decoding for Convolutional Codes}\label{sec:DecodingofConvolutional Codes}

A convolutional encoder can be represented by a 
trellis diagram
with the states $S_t=(u_{t-1},\hdots, u_{t-\mu})$ at each time unit $t$ representing the previous $\mu$ inputs.
Transitions occur exclusively between states $S_t$ at time unit $t$ and states $S_{t+1}$ at time unit $t + 1$;
see Figure \ref{trellis}.

The algorithm that follows is the Viterbi algorithm, widely known as ``dynamic programming'' applied to a trellis diagram \cite{forney2005}. This algorithm will be adapted for $k$-partial simplex convolutional codes
in Section \ref{sec:NewDecodingAlgorithm} to obtain a reduced complexity.

\begin{algorithm}[H]
\caption{Viterbi Decoding}\label{Viterbi}
Let $\mathcal{C}$ be a convolutional code with minimal generator matrix $G(z)$ 
and memory $\mu$. Let $r(z) = \sum_{i=0}^N r_i z^i$ be a received message
to be decoded, $N$ the input length.
Set up the trellis.\\
\textbf{Step 1:} Start with time $t = \mu$ and determine the partial path metric
$d(r_{[0,t-1]},c_{[0,t-1]})= \sum_{i=0}^{t-1}d(r_{i}, c_{i})$
for the single path entering each state of the trellis.
Save the path (called the survivor) and its metric for each
state.\\
\textbf{Step 2:}
\begin{algorithmic}
\For{$t$ from $\mu + 1$ to $N+1$}
\State \multiline{
For each state and the corresponding $2^k$ incoming paths do the following:}
\State \multiline{
1) Add the branch metric $d(r_{t-1},c_{t-1})$ entering the state to the partial path metric of the corresponding survivor at time $t-1$.\\
2) Compare the partial path metrics of all $2^{k}$ paths entering each state.\\
3) For each state, save the path with the smallest partial path metric (the survivor) and its metric $d(r_{[0,t-1]},c_{[0,t-1]})$.
If there are multiple such paths take one randomly.
Eliminate all other paths.
}
\EndFor
\end{algorithmic}
\textbf{Step 3:} 
Return the single survivor path and its metric.
\end{algorithm}

\begin{figure}
    \centering
    \tikzstyle{state}=[shape=circle,draw=black,scale=0.7]
\tikzstyle{lightedge}=[<-, dashed,>=stealth, scale=0.7]
\tikzstyle{mainstate}=[state,very thick]
\tikzstyle{mainedge}=[<-,very thick,>=stealth, scale=0.7]
\begin{tikzpicture}[]
\node               at (0,4.3) {$\scriptstyle{t=0}$};
\node[mainstate] (s1_1) at (0,3.8) {$00$};
\node               at (2,4.3) {$\scriptstyle{ t=1}$};
\node[state] (s1_2) at (2,3.8) {$00$}
    edge[lightedge] (s1_1);
\node[mainstate] (s3_2) at (2,2.4) {$10$}
     edge[mainedge] node[midway, above, yshift = -1mm, sloped] {\bf{\tiny{1/1111}}}(s1_1);
\node               at (4,4.3) {$\scriptstyle{ t=2}$};
\node[state] (s1_3) at (4,3.8) {$00$}
    edge[lightedge]  (s1_2);
\node[mainstate] (s2_3) at (4,3.1) {$01$}
    edge[mainedge] node[midway, above, xshift=-1mm, yshift = -1mm, sloped] {\bf{\tiny{0/0101}}}(s3_2);
\node[state] (s3_3) at (4,2.4) {$10$}
    edge[lightedge] (s1_2);   
\node[state] (s4_3) at (4,1.7) {$11$}
    edge[lightedge] (s3_2);
\node               at (6,4.3) {$\scriptstyle{ t=3}$};
\node[state] (s1_4) at (6,3.8) {$00$}
    edge[lightedge]  (s1_3)
    edge[lightedge]  (s2_3);
\node[state] (s2_4) at (6,3.1) {$01$}
    edge[lightedge] (s3_3)
    edge[lightedge] (s4_3);
\node[mainstate] (s3_4) at (6,2.4) {$10$}
    edge[lightedge] (s1_3)
    edge[mainedge] node[midway, above, yshift = -1mm, xshift=-1mm, sloped] {\bf{\tiny{1/1100}}}(s2_3);
\node[state] (s4_4) at (6,1.7) {$11$}
    edge[lightedge] (s3_3)
    edge[lightedge] (s4_3);
\node               at (8,4.3) {$\scriptstyle{ t=4}$};
\node[state] (s1_5) at (8,3.8) {$00$}
    edge[lightedge]  (s1_4)
    edge[lightedge]  (s2_4);
\node[state] (s2_5) at (8,3.1) {$01$}
    edge[lightedge] (s3_4)
    edge[lightedge] (s4_4);
\node[state] (s3_5) at (8,2.4) {$10$}
    edge[lightedge] (s1_4)
    edge[lightedge] (s2_4);
\node[mainstate] (s4_5) at (8,1.7) {$11$}
    edge[mainedge] node[midway, above, yshift = -1mm, sloped] {\bf{\tiny{1/1010}}} (s3_4)
    edge[lightedge] (s4_4);
\node               at (10,4.3) {$\scriptstyle{ t=5}$};
\node[state] (s1_6) at (10,3.8) {$00$}
    edge[lightedge]  (s1_5)
    edge[lightedge]  (s2_5);
\node[mainstate] (s2_6) at (10,3.1) {$01$}
    edge[lightedge] (s3_5)
    edge[mainedge] node[midway, above, yshift = -1mm, xshift=-1mm, sloped] {\bf{\tiny{0/0110}}} (s4_5);
\node               at (12,4.3) {$\scriptstyle{ t=6}$};
\node[mainstate] (s1_7) at (12,3.8) {$00$}
    edge[lightedge]  (s1_6)
    edge[mainedge] node[midway, above, yshift = -1mm, xshift=-1mm, sloped] {\bf{\tiny{0/0011}}} (s2_6);
\end{tikzpicture}
    \caption{Trellis for Algorithm \ref{Viterbi} applied to Example \ref{exviterbi}.} 
    \label{trellis}
\end{figure}
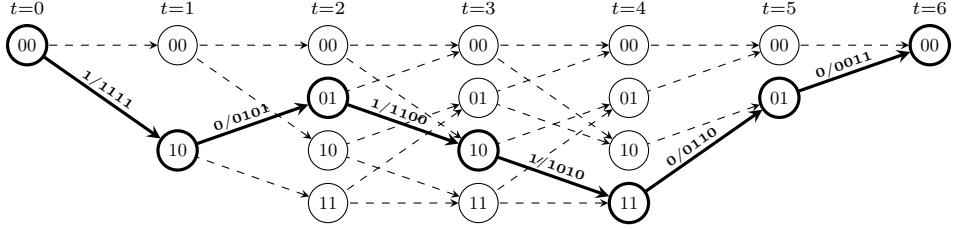

\begin{example}\label{exviterbi}
Take the $(4,1,2)$ convolutional code $\mathcal{C}\in\mathbb F_2[z]^4$ with generator matrix $G(z)=( 1\ \ 1+z\ \ 1+z^2\ \ 1+z+z^2)$.
As $d_{free}(\mathcal{C})=8$, 
$\mathcal{C}$ can correct up to $3$ errors.
Assume we want to send the message $u(z)=1+z^2+z^3$, which corresponds to the codeword $c(z)=(1111)+(0101)z+(1100)z^2+(1010)z^3+(0110)z^4+(0011)z^5$. The corresponding trellis is shown in Figure \ref{trellis}.
Assume we receive the word 
$r(z)=(1111)+(0101)z+(0100)z^2+(1010)z^3+(1111)z^4+(0011)z^5$ and want to decode it with the Viterbi algorithm.

According to Step 1 of the Viterbi Algorithm, we have to calculate the path metric $d=d(c_{[0,1]} , r_{[0,1]})$
for the $4$ paths ending in the $4$ states at time $t=2$.

State 00: 
$d=d(0000,1111)+d(0000,0101)=6$

State 01: 
$d=d(1111,1111)+d(0101,0101)=0$

State 10:
$d=d(0000,1111)+d(1111,0101)=6$

State 11:
$d=d(1111,1111)+d(1010,0101)=4$ \\
For $t=3$ two paths are ending in each state. For each input $u=(u_2,u_1,u_0)$ (i.e. $(u_2,u_1)$ is the current and $(u_1,u_0)$ the previous state)
we calculate the corresponding path metrics $d=d(c_{[0,2]} , r_{[0,2]})=d(c_{[0,1]} , r_{[0,1]})+d(c_2,r_2)$:

State 00: $u=(000), d=6+d(0000,0100)=6+\mathbf{1}=7$\\
         $  \indent\qquad \qquad\  u=(001), d=0+d(0011,0100)=0+\mathbf{3}=3$

State 01: 
$u=(010), d=6+d(0101,0100)=6+\mathbf{1}=7$\\
$ \indent\qquad \qquad\ u=(011), d=4+d(0110,0100)=4+\mathbf{1}=5$

State 10:
$u=(100), d=6+d(1111,0100)=6+\mathbf{3}=9$\\
$ \indent\qquad \qquad \ u=(101), d=0+d(1100,0100)=0+\mathbf{1}=1$

State 11:
$u=(110), d=6+d(1010,0100)=6+\mathbf{3}=9$\\
$ \indent\qquad \qquad \ u=(111), d=4+d(1001,0100)=4+\mathbf{3}=7$

i.e. here for all states the second path is the survivor. 

We continue in the same way for $t\geq 4$.
\end{example}

Note that in each time-step $t$ in the Viterbi algorithm one has to compute the distances between $r_{t-1}$ and all the codewords of the block code generated by  $(G_0^\top\ \cdots\ \ G_{\mu-1}^\top \ \tilde{G}_{\mu}^\top)^\top$.

\section{Decoding of Reed-Muller Codes}\label{sec:DecodingofHadamard}
In this section we describe the decoding of binary Reed-Muller codes $RM(1, m)$ following \cite{rushforth1969fast}.
These are $1$-partial simplex codes $\mathcal{S}(m+1)_1$.
Alternatively, we can describe them recursively.
A generator matrix for $RM(1, 1)$ is given by
\vspace{-1mm}
$$
    R(1) = 
    \begin{pmatrix}
        1 & 1\\
        0 & 1
    \end{pmatrix}.
$$
\vspace{-3mm}\\
Recursively we get a generator matrix for $RM(1, m+1)$ as
\begin{align}\label{Grec}
    R(m+1) = 
    \left[
		\begin{array}{c|c}
			R(m) & R(m)\\
			\hline
			0, \ldots, 0 & 1, \ldots, 1
		\end{array}
	\right]\in\mathbb F_2^{(m+2)\times 2^{m+1}}.
\end{align}
Next we explain how these codes are related to Hadamard matrices.
We define the Hadamard matrix $H_2 = (\begin{smallmatrix}
    1 & 1\\
    1 & -1
\end{smallmatrix})$
and 
recursively
$$
    H_{2^{m+1}} =
    \begin{pmatrix}
        H_{2^m} & H_{2^m} \\
        H_{2^m} & - H_{2^m}
    \end{pmatrix}.
$$

\begin{theorem}\label{order}
After replacing all entries in $H_{2^m}$ equal to $-1$ by $0$ , one obtains all rows of $H_{2^m}$  (in the correct order) if one forms linear combinations of the rows of $R(m)$ in the following way: Start with the first row of $R(m)$. Then, for each $i\in\{2,\hdots, 2^{m+1}\}$ add row $i$ to the already enumerated $2^{i-2}$ codewords in the same order.\\
Moreover if we add the first row of $R(m)$ to all these already obtained $2^{m}$ codewords in the same order, we obtain all rows of $-H_{2^m}$ replacing $-1$ by $0$ (in the correct order).\\
This implies that all codewords of $RM(1,m)$ can be represented by the rows of $\begin{pmatrix}
        H_{2^{m}} \\ - H_{2^{m}}
    \end{pmatrix}$.
\end{theorem}

\begin{proof}
We do induction with respect to $m$.
For $m=1$, we start with the codeword $(1, 1)$ and get the first row of $H_2$.
    Then we add the second row of $R(1)$ and get $(1, 0)$ which corresponds to the second row of $H_2$.
    Adding $(1, 1)$ to both codewords we get $- H_2$.\\
Assume we enumerated all linear combinations of the rows of $R(m)$ that include the first row as described in the theorem to obtain all rows of $H_{2^m}$. We will use \eqref{Grec} to prove that the statement of the theorem is true for $m+1$.
Obviously, taking linear combinations of $[R(m)\ R(m)]$ in the same way, we obtain all rows of $[H_{2^m}\ H_{2^m}]$ in the correct order. Now, we add the last row of $R(m+1)$, i.e.  $(0, \ldots, 0, 1, \ldots, 1)$, to all previously enumerated $2^{m-1}$ codewords of $RM(1,m+1)$ and obtain (after replacing $0$ by $-1$) all rows of $[H_{2^{m}}\ -H_{2^{m}}]$ in the correct order. Hence, in total, we obtained all rows of $H_{2^{m+1}} =     \begin{pmatrix}
        H_{2^m} & H_{2^m} \\
        H_{2^m} & - H_{2^m}
    \end{pmatrix}$ (in the correct order).\\
Since the first row of $R(m+1)$ is equal to $(1, \ldots, 1)$, adding this row to all the $2^{m+1}$ previously enumerated codewords of $RM(1,m+1)$, we obtain all rows of the matrix \\
$\begin{pmatrix}
        H_{2^{m+1}} \\ - H_{2^{m+1}}
    \end{pmatrix}=\begin{pmatrix}
        H_{2^m} & H_{2^m} \\
        H_{2^m} & - H_{2^m}\\
        -H_{2^m} & -H_{2^m} \\
        -H_{2^m} &  H_{2^m}
    \end{pmatrix}$.
\end{proof}

From now on we assume that the codewords $c$ of $RM(1,m)$ and received words $r$ have entries in $\lbrace -1, 1 \rbrace$.
Then $c_i = r_i$ if and only if $c_i \cdot r_i = 1$.
Otherwise this product is $-1$.
Thus, the closest codeword to a received word $r$ is the $c$ maximizing the inner product  $ r \cdot c = \sum_{i=1}^n r_i \cdot c_i.$
Note that in general
\begin{equation}\label{eq_distance_inner_product}
    d(r, c) = \frac{n - (r\cdot c)}{2},
\end{equation}
where $d(r, c)$ is the Hamming distance between $r$ and $c$ (which is the same also if we use $0$'s instead of $-1$'s).
Hence, we can go back and forth between Hamming distance and the value of the inner product.
In $RM(1,m)$, where all the codewords can be enumerated as the rows of $(\begin{smallmatrix}
    H_{2^m} \\
    - H_{2^m} 
\end{smallmatrix})$,
the codeword which is closest to the received word $r$ is the row of $(\begin{smallmatrix}
    H_{2^m} \\
    - H_{2^m}
\end{smallmatrix})$
that corresponds to the entry with maximal value of the vector resulting from the matrix vector product
$
    \begin{pmatrix}
        H_{2^m} \\
    - H_{2^m}
    \end{pmatrix}\cdot
    r^T.
$
We will use fast Hadamard multiplication, which essentially boils down to a decomposition of $H_{2^m}$ and allows to do the multiplication $H_{2^m}\cdot r^\top$ with $n \log(n)$ operations instead of $n^2$, which is the cost of a general matrix vector multiplication of this size.
One can show as in \cite{rushforth1969fast} that
$H_{2^m} = (CP)^m$,
where the permutation $P$ is the Faro-out-shuffle (see \cite{diaconis1983mathematics}) and
$$
    C = 
    \left(\begin{smallmatrix}
        H_2 & & 0\\
        & \ddots & \\
       0 & & H_2
    \end{smallmatrix}\right).
$$
Thus, a multiplication $C \cdot r^T$ needs $n = 2^m$ additions. Since we need $m$ such multiplications we get a total complexity of $2^m \cdot m = n \log(n)$.
We assume that permutations are constant time.
In sum we obtain
\begin{theorem}
    For a binary Reed-Muller code $RM(1, m)$ of length $n = 2^m$ the distance between all codewords and a received word can be computed with complexity $n \log(n)$.
    In particular decoding can be done with complexity $n \log(n)$.
\end{theorem}

Finally, note that comparing all codewords to the received word also has complexity $n^2$, which is equal to the complexity for the naive matrix vector multiplication without the fast Hadamard transform.

\section{Decoding of Partial simplex Codes}\label{sec:DecodingofPartialSimplexCodes}

In this section we generalize the results of the previous section 
to $k$-partial simplex codes with $k>1$.
Note that the generator matrix $S(k+\delta)_k$ can be written as

\vspace{-3mm}
\small\begin{equation*}
    S(k+\delta)_k = \left(
    \begin{array}{c|c|c|c|c}
        R(m)   & 0 \ldots 0            & 0 \ldots 0 & \ldots & 0 \ldots 0\\
                                & R(m-1)  & 0 \ldots 0 & \ldots & 0 \ldots 0\\
                                &                       & R(m-2) &\ldots & 0                                                  \ldots 0 \\
        \vdots                  & \vdots                & \vdots    &\vdots     & \vdots \\
                                &                       &           &        & R(\delta)
    \end{array}
    \right).
\end{equation*}
\normalsize
for $m = \delta + k - 1$.
The following example illustrates how to generalize the ideas from the previous section.
\begin{example}
    Take $k=3$ and $\delta=1$. Then, one obtains
    
    \vspace{-3mm}
    \small\begin{align*}
    &S(4)_3=\begin{pmatrix}
        R(3) & 0_{1\times 4} & 0_{1\times 2}\\
        & R(2) & 0_{1\times 2}\\
        & & R(1)\end{pmatrix}=\\
        &\left(\begin{array}{ccccccccccccccccccccccc}
       1 & 1 & 1 & 1 & 1 & 1 & 1 & 1& \vline & 0 & 0 & 0 & 0 & \vline &0 & 0\\ 0 & 1 & 0 & 1 & 0 & 1 & 0 & 1 & \vline & 1 & 1 & 1 & 1 & \vline &0 & 0\\
       0 & 0 & 1 & 1 & 0 & 0 & 1 & 1 & \vline &  0 & 1 & 0 & 1 & \vline &1 & 1\\
       0 & 0 & 0 & 0 & 1 & 1 & 1 & 1 & \vline & 0 & 0 & 1 & 1 & \vline &0 &1           \end{array}\right).
    \end{align*}
    \normalsize
We enumerate the codewords in the same way as explained in the previous section and replace $0$'s by $-1$'s.
We only have to check what happens to the two rightmost blocks.
Consider 
\vspace{-1mm}
\[
    \begin{pmatrix}
        0 & 0 & 0 & 0\\
        1 & 1 & 1 & 1\\
        0 & 1 & 0 & 1\\
        0 & 0 & 1 & 1
    \end{pmatrix}.
\]
\vspace{-2mm}\\
We start with the first row $(0, 0, 0, 0)$ and get $(-1, -1, -1, -1)$ by replacing $0$'s by $-1$'s.
Then we add the second row and get $(1, 1, 1, 1)$ which stays the same.
Adding $(0, 1, 0, 1)$ to the previous two vectors, we get $(0, 1, 0, 1)$ and $(1, 0, 1, 0)$ which correspond to $(-1, 1, -1, 1)$ and $(1, -1, 1, -1)$.
Next we add the last row $(0, 0, 1, 1)$ to all the previous vectors and get $(0, 0, 1, 1)$, $ (1, 1, 0, 0)$, $(0, 1, 1, 0)$, $(1, 0, 0, 1)$ which correspond to $(-1, -1, 1, 1)$, $(1, 1, -1, -1)$, $(-1, 1, 1, -1)$, $(1, -1, -1, 1)$.
Therefore the first half of the resulting matrix is
\vspace{-2mm}
$$
    \begin{pmatrix}
        -1 & 1 & -1 & 1 & -1 & 1 & -1 & 1\\
        -1 & 1 & 1 & -1 & -1 & 1 & 1 & -1\\
        -1 & 1 & -1 & 1 & 1 & -1 & 1 & -1\\
        -1 & 1 & 1 & -1 & 1 & -1 & -1 & 1
    \end{pmatrix}^T
$$
\vspace{-3mm}\\
Finally we add the first row $(0, 0, 0, 0)$ to all the previous $8$ vectors which gives us the same $8$ vectors in the same order.

Now let us consider the rightmost block of $S(4)_3$.
We start with the first row $(0, 0)$ corresponding to $(-1, -1)$.
Then we add the second row, which is $(0, 0)$, to get $(0, 0)$ corresponding to $(-1, -1)$.
After this we add $(1, 1)$ to the previous vectors and get $(1, 1), (1, 1)$ which remain unchanged under replacing $0$'s by $-1$'s.
We continue by adding $(0, 1)$ to the previous four vectors which gives $(0, 1), (0, 1), (1, 0), (1, 0)$ corresponding to $(-1, 1), (-1, 1), (1, -1), (1, -1)$.
Therefore the first half of the resulting matrix is
\vspace{-2.5mm}
\[
\begin{pmatrix}
    -1 & -1 & 1 & 1 & -1 & -1 & 1 & 1\\
    -1 & -1 & 1 & 1 & 1 & 1 & -1 & -1
\end{pmatrix}^T.
\]
\vspace{-4mm}\\
Then we add the first row $(0, 0)$ to all the previous $8$ vectors and get the same $8$ vectors again.
Putting all three blocks together we get an enumeration of all codewords where we have replaced $0$'s by $-1$'s.
\end{example}
The next theorem formalizes the idea of the previous example and describes the main step for the decoding of $k$-partial simplex codes with $k>1$.

\begin{theorem}
\label{thm_enumeration_partial_simplex_codes}
    Enumerating all the codewords of $\mathcal{S}(\delta+k)_k$ in the same order as done for $RM(1,m)$ in Theorem \ref{order}, setting $m = k + \delta - 1$ and replacing $0$'s by $-1$'s, we get the matrix
     \begin{align*}
        \tilde{H}_{2^m} = 
        \begin{pmatrix}
            H_{2^m} & \hat{H}_{2^{m-1}} & \cdots & \hat{H}_{2^\delta} \\
            -H_{2^m} & \hat{H}_{2^{m-1}} & \cdots & \hat{H}_{2^\delta}
        \end{pmatrix}\quad\text{with}\quad
      \hat{H}_{2^{m-l}} = T\cdot 
       \begin{pmatrix}
        -H_{2^{k + \delta - l-1}}\\
        \vdots\\
        -H_{2^{k + \delta - l-1}}\\
        H_{2^{k + \delta - l-1}}\\
        \vdots\\
        H_{2^{k + \delta - l-1}}
    \end{pmatrix}  
   \end{align*}
       where $-H_{2^{k + \delta - l - 1}}$ and $H_{2^{k + \delta - l - 1}}$ occur $2^{l-1}$ times and $T$ is given by the following permutation

\begin{align*}T \cdot \begin{pmatrix}
    x_0\\
    x_1\\
    \vdots\\
    x_{2^{k + \delta - 1} - 2}\\
    x_{2^{k + \delta - 1} - 1}
\end{pmatrix} = 
\begin{pmatrix}
    x_0\\
    x_{2^{k + \delta - l - 1}}\\
    x_{2 \cdot 2^{k + \delta - l - 1}}\\
    x_{3 \cdot 2^{k + \delta - l - 1}}\\
    \vdots\\
   x_{2^{k + \delta -1} - 2^{k + \delta - l - 1}}\\
    x_1\\
    x_{1 + 2^{k + \delta - l - 1}}\\
    x_{1 + 2 \cdot 2^{k + \delta - l - 1}}\\
    x_{1 + 3 \cdot 2^{k + \delta - l - 1}}\\
    \vdots\\
    x_{2^{k + \delta - 1} - 2^{k + \delta - l - 1} + 1}\\
    \vdots
\end{pmatrix}
\end{align*}
\end{theorem}

\begin{proof}
    Using the same procedure as for $RM(1,\delta+k-1)$ in Theorem \ref{order}, the $2^{\delta + k - 1}$ left-most columns of the codewords of $S(\delta+k)_k$ correspond to $\begin{pmatrix}
        H_{2^{\delta + k - 1}}\\
        - H_{2^{\delta + k - 1}}
    \end{pmatrix}$.
Now we will investigate the remaining columns of the matrix blockwise, i.e., for $l=1,\hdots, k-1$, we consider the block
\[
    \begin{pmatrix}
        0_{l\times 2^{k+\delta-1-l}}\\
        R(k + \delta - l - 1)
    \end{pmatrix}\in\mathbb F_2^{(k+\delta)\times 2^{k+\delta-1-l}},
\]
where the first $l$ rows are all zero.
If we do the same procedure as we did on $R( k + \delta - 1)$, we get the following. 

We start with the first row, which is $(0, \ldots, 0)$.
 For $i \in \lbrace 2\ldots, l \rbrace$, adding row $i$ to all $2^{i-2}$ previously enumerated codewords gives $2^{l-1}$ times the zero vector, which corresponds to $(-1, \ldots, -1)$.
Next with row $l + 1$ we add $(1, \ldots, 1)$ to all previous codewords, which gives $2^{l-1}$ times $(1, \ldots, 1)$.

We prove inductively that for $j=1,\hdots,k+\delta-l$ after adding row $l+j$ to the previously enumerated rows, we have enumerated the first $2^{j-1}$ rows of $- H_{2^{k + \delta - l - 1}}$ interleaved with the first $2^{j-1}$ rows of $H_{2^{k + \delta - l - 1}}$ in blocks of length $2^{l-1}$.
There will be a block of $2^{l-1}$ times a row of $- H_{2^{k + \delta - l - 1}}$. then a block of $2^{l-1}$ times the same row of $H_{2^{k + \delta - l - 1}}$.
Then it continues with the next row of $- H_{2^{k + \delta - l - 1}}$.
We already showed the statement for $j = 1$.
Adding row $j$ of $R(k + \delta - l - 1)$ to all previous codewords we get again blocks of length $2^{l-1}$.
Since the blocks are in the same order apart from the interleaving the inductive step follows.
Adding the first row of the block, i.e. $(0, \ldots, 0)$, to all previously obtained codewords in the same order, we obtain again $\hat{H}_{2^{k + \delta - l - 1}}$. This is the reason why $\tilde{H}_{2^m}$ contains this matrix twice.
The matrix $\hat{H}_{2^{k + \delta - l - 1}}$ obtained in this way is a shuffle of the matrix

\[
    \begin{pmatrix}
        -H_{2^{k + \delta - l-1}}\\
        \vdots\\
        -H_{2^{k + \delta - l-1}}\\
        H_{2^{k + \delta - l-1}}\\
        \vdots\\
        H_{2^{k + \delta - l-1}}
    \end{pmatrix}  
\]

with $2^{l-1}$ times $-H_{2^{k + \delta - l - 1}}$ and $2^{l-1}$ times $H_{2^{k + \delta - l - 1}}$.
This shuffle is given by the permutation $T$ from the theorem.
\end{proof}

Theorem \ref{thm_enumeration_partial_simplex_codes} provides a way to decode partial simplex codes efficiently
as explained in the following.
Set $m=\delta + k - 1$.
We split the received word $r$ into parts $(r^{(m)}, \ldots, r^{(\delta)})$ where $r^{(i)}$ has length $2^i$.
Note that forming the inner product of $r$ with all codewords of $\mathcal{S}(\delta+k)_k$ results in the multiplication
\begin{align*}
    \tilde{H}_{2^m} \cdot r^T \hspace{-1mm} &= 
    \begin{pmatrix}
        H_{2^m} & \hat{H}_{2^{m-1}} & \cdots & \hat{H}_{2^\delta} \\
        -H_{2^m} & \hat{H}_{2^{m-1}} & \cdots & \hat{H}_{2^\delta}
    \end{pmatrix}
    \cdot (r^{(m)}, \ldots, r^{(\delta)})^T\\
   &  = 
    \begin{pmatrix}
        H_{2^m} \cdot r^{(m), T} + \hat{H}_{2^{m-1}} \cdot r^{(m-1), T} + \cdots + \hat{H}_{2^\delta} \cdot r^{(\delta), T}\\
        -H_{2^m} \cdot r^{(m), T} + \hat{H}_{2^{m-1}} \cdot r^{(m-1), T} + \cdots + \hat{H}_{2^\delta} \cdot r^{(\delta), T}
    \end{pmatrix}.
\end{align*}
With Theorem \ref{thm_enumeration_partial_simplex_codes} we see that by copying the results from the multiplications and using the permutation $T$, it is enough for the comparison of
$r$ with all codewords of $\mathcal{S}(\delta+k)_k$ to compute $H_{2^i} \cdot r^{(i), T}$ for each $i \in \lbrace \delta, \ldots, \delta + k - 1 \rbrace$ with the fast Hadamard multiplication.
The overall complexity is thus
$$
    \sum_{i = \delta}^m 2^i\cdot  i \leq 2^\delta \sum_{i = 0}^{m - \delta} 2^i\cdot m = 2^\delta\cdot  2^k \cdot (\delta + k - 1) \leq 2 n \log(n).
   $$
Therefore we can conclude:
\begin{theorem}
\label{thm_decoding_k_partial_simplex}
    For a $k$-partial simplex code $\mathcal{S}(\delta+k)_k$ one can compute the distance between a received word and all codewords with complexity
   $O(n \log(n)),$
      where $n = 2^{\delta + k} - 2^\delta$.
    In particular, decoding can be done with complexity $O(n \log(n))$.
\end{theorem}
 
\section{Improved Viterbi algorithm for a class of optimal convolutional codes}\label{sec:NewDecodingAlgorithm}

In this section we present a new algorithm capable of decoding $k$-partial simplex convolutional codes with lower complexity than the Viterbi algorithm.
We assume that we have replaced the $0$-entries of a received word by $-1$'s, i.e. $r = (r_0, \ldots, r_N) \in \lbrace -1, 1 \rbrace^{n \cdot (N+1)}$.
On the edges in the trellis between time $t-1$ and $t$, for each $t \in \lbrace \mu + 1, \ldots, N + \mu \rbrace$, there occurs each codeword of the $k$-partial simplex code exactly once.
Therefore the crucial task is to compare the distances between all the codewords of the $k$-partial simplex code and $r_{t-1}$.
We will use Theorem \ref{thm_decoding_k_partial_simplex} to do this efficiently.
The decoding process begins with precomputing the permutation $P$ for Hadamard multiplication.
Furthermore, we compute a permutation $Q$ as described in the next paragraph which allows us to go linearly through the stored distances in 1) in Step 2 of Algorithm \ref{alg:cap}.
This implies that the lookup happens with complexity $O(1)$. 

We built the trellis such that the states and the tuples $(u_{t-1},\hdots,\tilde{u}_{t-1-\mu})$ corresponding to the $2^{\mu+1}$ branches between the states $S_{t-1}$ and $S_t$ in the trellis are in lexicographical order starting with the smallest (see Example \ref{exviterbi}).
Moreover, we obtained the rows of $\tilde{H}_{2^m}$ by forming linear combinations of the rows of $S(\delta+k)_k$, i.e. by calculating $(u_{t-1},\hdots,\tilde{u}_{t-1-\mu})\cdot S(\delta+k)_k$ for all $(u_{t-1},\hdots,\tilde{u}_{t-1-\mu})\in\mathbb F_2^{\delta+k}$ in a fixed order.
Hence, if row $i$ of $\tilde{H}_{2^m}$ is the sum of the rows $j_1,\hdots,j_a$ of $S(\delta+k)_k$, then row $i$ of $Q$ contains the standard basis vector $e_j$, where $j$ corresponds to the lexicographical order of $e_{j_1}+\cdots+e_{j_a}$ in $\mathbb F_2^{\delta+k}$, i.e., for $1\leq i\leq 2^m$, $j=\sum_{w=1}^a 2^{m+1-j_w}+1$  and for $2^m+1\leq i\leq 2^{m+1}$, $j=\tilde{j}-2^m$ if $e_{\tilde{j}}$ is contained in row $i-2^m$ of $Q$.

Once this information is built up, the algorithm is ready to initialize the trellis and recreate the sequence of bits that were input to the convolutional encoder when the message was encoded for transmission.
This is accomplished by the steps given in Algorithm \ref{alg:cap}.

\begin{algorithm}
\caption{Improved Viterbi Algorithm}\label{alg:cap}
Let $\mathcal{C}$ be a $k$-partial simplex convolutional code with minimal generator matrix $G(z)$, memory $\mu$, received message $r = (r_0, \ldots, r_N) \in \lbrace -1, 1 \rbrace^{n \cdot (N+1)}$ to be decoded, input length $N$.
Precompute the permutation $Q$ and the permutation $P$ for fast Hadamard multiplication and set up the trellis.\\
{\textbf{Step 1:} Perform Step 1 of the Algorithm \ref{Viterbi}.}\\
\textbf{Step 2:}
\begin{algorithmic}
\For {$t$ from $\mu +1$ to  $N-\mu+1$}\State \multiline{Compute $Q \cdot \tilde{H}_{2^m}\cdot r_{t-1}^T$ and use Equation \eqref{eq_distance_inner_product} to compute the distances between $r_{t-1}$ and the codewords of $\mathcal{S}(\delta+k)_k$.
Save the result and 
for each state and the corresponding $2^k$ incoming paths do the following:
}
\State \multiline{
1) Add the branch metric $d(r_{t-1},c_{t-1})$ entering the state to the partial path metric of the corresponding survivor at $t-1$.
Do this by going linearly through the stored distances.\\
2) Compare the partial path metrics of all $2^{k}$ paths entering each state.\\
3) For each state, save the path with the smallest partial path metric (i.e. the survivor) and its metric $d(r_{[0,t-1]},c_{[0,t-1]})$.
If there are multiple such paths take one randomly.
Eliminate all other paths in the trellis.
}
\EndFor
\end{algorithmic}
\textbf{Step 3:} For $t$ from $N-\mu+2$ to $N+1$, do the steps of the for loop in Algorithm \ref{Viterbi}.\\
\textbf{Step 4:} Return the single survivor path and its metric.
\end{algorithm}

\begin{example}
We take the same generator matrix, codeword and received word as in Example \ref{exviterbi} to illustrate Algorithm 2. Step 1 ($t=2$) stays the same as in Algorithm 1.
Since $m=\delta+k-1=2$, for step $t=3$, 
after replacing $0$'s by $-1$'s, we obtain $r_2=(-1\ 1\ -1\ -1)$ and
\vspace{-1mm}
\begin{align*}
H_4 \cdot r_2 = (C \cdot  P)^2\cdot r_2 = (-2, -2, 2, -2)^T.
\end{align*}
\vspace{-4mm}\\
We get 
$Q\cdot\begin{pmatrix}
    H_4\\ -H_4
\end{pmatrix}\cdot r_2=\begin{pmatrix}2& \hspace{-0.2cm} -2 \hspace{-0.2cm}&\hspace{-0.2cm} 2& \hspace{-0.2cm} 2 \hspace{-0.2cm}& \hspace{-0.2cm} -2 \hspace{-0.2cm}& \hspace{-0.2cm} 2 \hspace{-0.2cm}& \hspace{-0.2cm} -2 \hspace{-0.2cm}&  \hspace{-0.2cm}-2\end{pmatrix}^\top
$
with
$$Q=\begin{pmatrix}
    e_5^\top & \hspace{-0.2cm}
     e_7^\top & \hspace{-0.2cm}
    e_6^\top& \hspace{-0.2cm}
     e_8^\top& \hspace{-0.2cm}
     e_1^\top& \hspace{-0.2cm}
     e_3^\top& \hspace{-0.2cm}
   e_2^\top& \hspace{-0.2cm}
     e_4^\top\hspace{-0.2cm}
\end{pmatrix}^\top.$$
Using Equation \eqref{eq_distance_inner_product} we obtain $\begin{pmatrix} 1 \hspace{-0.2cm}& \hspace{-0.2cm} 3 \hspace{-0.2cm}& \hspace{-0.2cm} 1 \hspace{-0.2cm}& \hspace{-0.2cm} 1 \hspace{-0.2cm}& \hspace{-0.2cm} 3 \hspace{-0.2cm}& \hspace{-0.2cm}1 \hspace{-0.2cm} & \hspace{-0.2cm}3 \hspace{-0.2cm}& \hspace{-0.2cm}3 \hspace{-0.2cm}\end{pmatrix}^\top$ as summands for the path metrics corresponding to the bold values in Example \ref{exviterbi}.
\end{example}

\section{Complexity Analysis}\label{sec:ComplexityAnalysis}
In this section we give the complexity analysis of Algorithm \ref{alg:cap}.
We assume that arithmetic operations and in particular additions, subtractions, multiplications and comparisons are constant time, i.e., $O(1)$.
As long as the parameters $n$ and $N$ are not too large the assumption that arithmetic operations are $O(1)$ reflects the CPU time more accurately than the bit complexity.
We show that under this assumption our algorithm outperforms the classical Viterbi algorithm.

Recall that the fast Hadamard transform has complexity $n \log(n)$.
We only consider the complexity for a fixed time in $\mu+1, \ldots, N-\mu+1$.
As long as $N$ is large in comparison to $\mu$ the final complexity is essentially $N$ times the complexity $C$ we get for one timestep, i.e., 
$N \cdot C.$
First we use Theorem \ref{thm_decoding_k_partial_simplex} to compute the distances between the codewordes of $\mathcal{S}(\delta + k)_k$ and $r_{t-1}$ together with a permutation and get a complexity of $O(n \log(n))$.
Next we have a for loop over all states in the order we have described above.
So we get $2^\delta$ times the cost $C_S$ for each iteration of the for loop over the states, i.e., $C = O(n \log(n) + 2^\delta \cdot C_S)$.
For the cost $C_S$ of a single loop we have to consider all the $2^k$ incoming edges.
Then we look at the paths which end in the starting state of the edge and compute the sum of path metric and hamming weight corresponding to the edge.
Thanks to the ordering we can go linearly over everything, so the cost is $2^k$.
Then we store the new path with its new metric.
Thus, the overall cost at a fixed time is
$C = O(n \log(n) + 2^\delta C_S) = O(n \log(n) + 2^\delta \cdot 2^k) = O(n \cdot \log(n))$.
Finally, including the time steps we get:
\begin{theorem}
The complexity of Algorithm \ref{alg:cap} for $k$-partial simplex convolutional codes is
    $O(N \cdot n \log(n)).$
\end{theorem}

Compared to this the Viterbi algorithm has complexity $O(N \cdot 2^\delta \cdot 2^k \cdot n) = O(N \cdot n^2)$.
The factor $n$ comes from the comparison of $n$ bits between $r_{t-1}$ and a codeword of the underlying blockcode which we can avoid 
using the fast Hadamard multiplication.
Also note that the complexity for sequential decoding, that was introduced by Wozencraft in \cite{wozencraft}, is at best $O(N\cdot n)$, which is almost achieved by Algorithm \ref{alg:cap}. Furthermore sequential decoding is suboptimal, i.e., it not always returns the closest codeword.

\section{Conclusion}
We used the fast Hadamard multiplication for the decoding of binary first-order Reed-Muller codes to obtain an improved Viterbi algorithm for partial simplex convolutional codes.
Interesting problems for future work are to generalize the ideas to other classes of codes and to investigate whether it is also possible to significantly improve sequential decoding algorithms.

\section*{Acknowledgment}
This work is supported by the SNSF grant n. 212865, by CIDMA through FCT, \url{https://doi.org/10.54499/UIDB/04106/2020},
\url{https://doi.org/10.54499/UIDP/04106/2020} and by FCT grant UI/BD/151186/2021.

\bibliographystyle{IEEEtran}
\bibliography{literature}

\end{document}